\documentclass[a4paper,english,11pt]{article}
\usepackage{amsmath}
\usepackage{amssymb}
\usepackage{amsthm}
\usepackage{tikz}
\usepackage{algorithm}
\usepackage{algorithmic}
\usepackage{pgfplots}
\usepackage{fullpage}

\usetikzlibrary{decorations.pathmorphing}
\usetikzlibrary{positioning}
\usetikzlibrary{arrows,calc}

\tikzstyle{sommet}=[circle, draw, fill=black, inner sep=2pt]
\tikzstyle{opt}=[draw = red, dashed]
\tikzstyle{alg}=[draw = blue]
\tikzstyle{algopt}=[double]
\tikzstyle{none}=[dotted]
\tikzstyle{Snake}=[decorate, decoration={snake}]
\tikzset{
>=stealth',
help lines/.style={dashed, thick},
axis/.style={<->},
important line/.style={thick},
connection/.style={thick, dotted},
}

\usepackage{graphicx}

\usepackage[textsize=tiny]{todonotes}

\newcommand{\comp}{\textrm{cr}}
\newcommand{\br}{\textrm{dr}}

\newtheorem{theorem}{Theorem}
\newtheorem{definition}{Definition}
\newtheorem{lemma}{Lemma}

\newtheorem{corollary}[lemma]{Corollary}

\newtheorem*{Remarks*}{Remarks}
\usepackage{microtype}%if unwanted, comment out or use option "draft"

\title{Best-of-two-worlds analysis of online search}

\author{Spyros Angelopoulos\thanks{Sorbonne Universit\'e, CNRS, Laboratoire d'informatique de Paris 6, LIP6, F-75252 Paris, France. Supported by ANR OATA, DIM RFSI DACM, and the  Fondation Mathématique Jacques Hadamard  Programme Gaspard Monge PGMO.}
 \and Christoph D\"urr\footnotemark[1]{} \and Shendan Jin\footnotemark[1]{}}

\begin{document}

\maketitle

\begin{abstract}
In search problems, a mobile searcher seeks to locate a target that hides in some unknown position of the environment. Such problems are typically considered to be of an on-line nature, in that the input is unknown to the searcher, and the performance of a search strategy is usually analyzed by means of the standard framework of the competitive ratio, which compares the cost incurred by the searcher to an optimal strategy that knows the location of the target. However, one can argue that even for simple search problems, competitive analysis fails to distinguish between strategies which, intuitively, should have different performance in practice. 

Motivated by the above, in this work we introduce and study measures {\em supplementary} to competitive analysis in the context of search problems. In particular, we focus on the well-known problem of {\em linear search}, informally known as the {\em cow-path problem}, for which there is an infinite number of strategies that achieve an optimal competitive ratio equal to 9. We propose a measure that reflects the rate at which the line is being explored by the searcher, and which can be seen as an extension of the {\em bijective ratio} over an uncountable set of requests. Using this measure we show that a natural  strategy that explores the line aggressively is optimal among all 9-competitive strategies. This provides, in particular, a strict separation from the competitively optimal doubling strategy, which is much more conservative in terms of exploration. We also provide evidence that this aggressiveness is requisite for optimality, by showing that any optimal strategy must mimic the aggressive strategy in its first few explorations. 

\end{abstract}

\section{Introduction}
\label{sec:introduction}

Searching for a hidden target is an important paradigm in computer science and operations research, with numerous applications. A typical search problem involves an environment, a mobile {\em searcher} (who may, or may not, have knowledge of the environment) and a {\em hider} (sometimes also called {\em target}) who hides at some position within the environment that is oblivious to the searcher. The objective is to define a {\em search strategy}, i.e., a traversal of the environment, that optimizes a certain efficiency criterion. A standard approach to the latter is by means of {\em competitive analysis}, in which we seek to minimize the worst-case cost for locating the target, divided by some concept of ``optimal'' solution; e.g., the minimum cost to locate the target once its position is known. Even prior to the advent of online computation and competitive analysis, search games had already been studied under such {\em normalized} measures within operations research~\cite{beck:yet.more}. Explicit studies of the competitive ratio and the closely related {\em search ratio} were given in~\cite{yates:plane} and~\cite{koutsoupias:fixed}, respectively, and led to the development of {\em online searching}~\cite{jaillet:online,onlinesearching.survey} as a subfield of online computation. See also~\cite{searchgames} for an in-depth treatment of search games, including the role of payoff functions that capture the competitive ratio.

In this work we revisit one of the simplest, yet fundamental search problems, namely the {\em linear search}, or, informally, {\em cow-path} problem. The setting involves  an infinite (i.e., unbounded) line, with a point $O$ designated as its origin, a searcher which is initially placed at the origin, and an immobile target  which is at some position on the line that is unknown to the searcher. More specifically, the searcher does not know whether the hider is at the left branch or at the right branch of the line. The searcher's {\em strategy} $S$ defines its exploration of the line, whereas the hider's strategy $H$ is determined by its placement on the line. Given strategies $S,H$, the {\em cost} of locating the hider, denoted by $c(S,H)$ is the total distance traversed by the searcher at the first time it passes over $H$. Let $|H|$ denote the distance of the hider from the origin. The competitive ratio of $S$, denoted by $\comp(S)$, is the worst-case normalized cost of $S$, among all possible hider strategies. Formally, 
\begin{equation}
\comp(S)=\sup_H \frac{c(S,H)}{|H|}.
\label{eq:competitive.ratio}
\end{equation}    

% A strategy $\Sigma$ for the searcher can be formally described as an infinite sequence of lengths of the form 
% \[
% \Sigma=(x_o^{d_0}, x_1^{d_1}, x_2^{d_2}, \ldots),
% \]
% where for each $i \in {\mathbb Z}$, $x_i>0$ (the length of the $i$-th search segment, and $d_i \in \{{\tt left,right}\}$ is the direction of the search 
% (to the left or to the right of the search). More specifically, for the $i$-th search segment, the searcher starts at the origin, searches the ray 
% $d_i$ up to a length $x_i$, and returns to $O$.  

% Without loss of generality, we can assume that if for a given $i$, $d_i=d_{i+1}$, then $x_i>x_{i+1}$,
% otherwise the $(i+1)$-th search segment is useless. Moreover, we can assume that $d_i \neq d_{i+1}$, since the $i$-th search segment can be omitted. Thus, a search strategy can be described by the sequence $\Sigma=(x_o, x_1, x_2, \ldots)$, such that consecutive segments search different rays, and such that $x_i < x_{i+2}$, for all $i$. 

It has long been known~\cite{beck:ls, gal:general} that the competitive ratio of linear search is 9, and is achieved by a simple {\em doubling} strategy: in iteration $i$, the searcher starts from $O$, explores branch $i \bmod 2$ at a length equal to $2^i$, and then returns to $O$. However, this strategy is not uniquely optimal; in fact, it is known that there is an infinite number of competitively optimal strategies for linear search (see Lemma~\ref{lem:CR_9_inequalities} in Section~\ref{sec:bijective.comp}). In particular, consider an ``aggressive'' strategy, which in each iteration searches a branch to the maximum possible extent, while maintaining a competitive ratio equal to 9. This can be achieved by searching, in iteration $i$, branch $i \bmod 2$ to a length equal to $(i+2)2^{i+1}$ (see Corollary~\ref{cor:R4_xn_and_Tn}).

While both {\tt doubling} and {\tt aggressive} are optimal in terms of competitive ratio, there exist realistic situations in which the latter may be preferable to the former. Consider, for example, a search-and-rescue mission for a missing backpacker who has disappeared in one of two (very long) concurrent, hiking paths. Assuming that we select our search strategy from the space of 9-competitive strategies, it makes sense to choose one that is tuned to discovering new territory, rather than a conservative strategy that tends to often revisit already explored areas.

With the above observation in mind, we first need to quantify what constitutes efficiency in exploration. To this end, given a strategy $S$ and $l \in {\mathbb R}^+$, we define $D(S,l)$ as the cost incurred by $S$ the first time the searcher has explored an aggregate length equal to $l$, combined in both branches. An efficient strategy should be such that $D(S,l)$ is small, for all $l$. Unfortunately, this criterion by itself is insufficient: Consider a strategy that first searches one branch to a length equal to $L$, where $L$ is very large. Then $D(S,l)$ is as small as possible for all $l<L$; however, this is hardly a good strategy, since it all but ignores one of the branches (and thus its competitive ratio becomes unbounded as $L \rightarrow \infty$).

To remedy this situation, we will instead use the above definition in a way that will allow us a pairwise comparison of strategies, which also considers all possible explored lengths. More formally, we define the following:
\begin{definition}
Let $S_1,S_2$ denote two search strategies, we define the {\em discovery ratio of $S_1$ against $S_2$}, denoted by $\br(S_1,S_2)$,  as
\[
\br(S_1,S_2)=\sup_{l \in {\mathbb R}^+} \frac{D(S_1,l)}{D(S_2,l)}.
\]
Moreover, given a class ${\cal S}$ of search strategies, the {\em discovery ratio of $S$ against the class ${\cal S}$} is defined as 
\[
\br(S, {\cal S})=\sup_{S' \in {\cal S}} \br(S,S').
\]
In the case ${\cal S}$ is the set $\Sigma$ of all possible strategies, we simply call $\br(S, {\cal S})$ the {\em discovery ratio of $S$},
and we denote it by $\br(S)$. 
\label{def:br}
\end{definition}

Intuitively, the discovery ratio preserves the worst-case nature of competitive analysis, and at the same time bypasses the need for 
an ``offline optimum'' solution. Note that if a strategy $S$ has competitive ratio $c$ then it also has discovery ratio
$c$; this follows easily from the fact that for every hider position $H$, $c(S,H) \geq D(S,|H|)$. 
However, the opposite is not necessarily true.

It is worth pointing out that the discovery ratio can be interpreted, roughly speaking, as the {\em bijective ratio} over a continuous space
of request sequences, or, more precisely, as its analogue in the setting in which this space consists of an infinite, uncountable set of requests. 
The bijective ratio was introduced in~\cite{bijective.ratio} as an extension of (exact) bijective analysis of online algorithms~\cite{ADLO07:paging},
and which in turn is based on the pairwise comparison of the costs induced by two online algorithms over all request sequences of a certain size.
Bijective analysis has been applied in fundamental online problems (with a discrete, finite set of requests) such as paging and list update~\cite{AS:bijective}, $k$-server~\cite{Boyar:measures,bijective.ratio}, and online search\footnote{In~\cite{DBLP:journals/tcs/BoyarLM14}, online search refers to the problem of selling a specific item at the highest possible price, and is not related to the problem of searching for a target.}~\cite{DBLP:journals/tcs/BoyarLM14}. Our interpretation of the bijective ratio can be useful for other online problems which are defined over a continuous setting of requests (e.g., $k$-server problems defined over a metric space rather than over a finite graph).  

The above observation implies that the discovery ratio inherits the appealing properties of bijective analysis, which further motivate its choice. 
In particular, note that bijective analysis has helped to identify theoretically efficiently algorithms which also tend to perform well in practice (such as Least-Recently-Used for paging~\cite{AS:bijective}, and greedy-like $k$-server policies for certain types of metrics~\cite{bijective.ratio}). Furthermore, if an algorithm has bijective ratio $c$, then its average cost, assuming a uniform distribution over all request sequences of the same length, is within a factor $c$ of the average cost of any other algorithm. Thus, bijective analysis can be used to establish  ``best of both worlds''  types of performance comparisons. In fact, assuming again uniform distributions, much stronger conclusions can be obtained, in that bijective analysis implies a {\em stochastic dominance} relation between the costs of the two algorithms~\cite{bijective.ratio}. 

It should be noted that the central question we study in this work is related to a phenomenon that is not unusual  in the realm of online computation.
Namely, for certain online problems, competitive analysis results in very coarse performance classification of algorithms. This is due to the pessimistic, worst-case nature of the competitive ratio.  The definitive example of an online problem in which this undesired situation occurs is the (standard) paging problem in a virtual memory system, which motivated the introduction of several analysis techniques alternative to the competitive ratio (see~\cite{survey} for a survey). In our paper we demonstrate that a similar situation arises in the context of online search, and we propose a remedy by means of the discovery ratio. We emphasize, however, that in our main results, we apply the discovery ratio as supplementary to the competitive ratio, instead of using it antagonistically as a measure that replaces the competitive ratio altogether. 

\section{Connections between the discovery and the bijective ratios}
\label{sec:motivation}

In this section we establish a connection between the discovery and the bijective ratios.
Bijective analysis was introduced in~\cite{ADLO07:paging} in the context of online computation, assuming that each request is drawn from a discrete, finite set. For instance, in the context of the paging problem, each request belongs to the set of all pages. Let ${\cal I}_n$ denote the set of all 
requests of size $n$. For a cost-minimization problem $\Pi$ with discrete, finite requests, let $\pi:{\cal I}_n \rightarrow {\cal I}_n$ denote a bijection over ${\cal I}_n$. Given two online algorithms $A$ and $B$ for $\Pi$, the bijective ratio of $A$ against $B$, is defined as 
\[
\textrm{br}(A,B)=\min_{\pi:{\cal I}_n \rightarrow {\cal I}_n} \sup_{\sigma \in {\cal I}_n} \frac{A(\sigma)}{B(\pi(\sigma))}, \textrm{ for all $n\geq n_0$},
\]
where $A(\sigma)$ denotes the cost of $A$ on request sequence $\sigma$.

Assuming ${\cal I}_n$ is finite, 
an equivalent definition of $\textrm{br}(A,B)$ is as follows. Let $A(i,n)$ denote the $i$-th least costly request sequence for $A$ among request sequences in 
${\cal I}_n$. Then 
\[
\textrm{br}(A,B) =\sup_n \max_i \frac{A(i,n)}{B(i,n)}.
\]
Consider in contrast, the linear search problem. Here, there is only one request: the unknown position of the hider (i.e., $n=1$). However, the set of all requests is not only infinite, but uncountable. Thus, the above definitions do not carry over to our setting, and we need to seek alternative definitions. One possibility is to discretize the set of all requests (as in~\cite{bijective.ratio}). Namely, we may assume that the hider can hide only at integral distances from the origin. Then given strategies $S_1,S_2$, one could define the bijective ratio of $S_1$ against $S_2$ as $\sup_i \frac{S_1(i)}{S_2(i)}$, where $S(i)$ denotes the $i$-th least costly request (hider position) in strategy $S$.

While the latter definition may indeed be valid, it is still not a faithful representation of the continuous setting. For instance,
for hiding positions ``close'' to the origin, the discretization adds overheads that should not be present, and skews the expressions of the ratios. 
For this reason, we need 
to adapt the definition so as to reflect the continuous nature of the problem. Specifically, note that while the concept ``the cost of the $i$-th least costly request in $S$'' is not well-defined in the continuous setting, the related concept of ``the cost for discovering a total length equal to $l$ in $S$'' is, in fact, well defined, and is precisely the value $D(S,l)$.  We can thus define the bijective ratio of $S_1$ against $S_2$ as
\[
\textrm{br}(S_1,S_2) =\sup_l \frac{D(S_1,l)}{D(S_2,l)},
\]
which is the same as the definition of the discovery ratio (Definition~\ref{def:br}).

% Unlike previously studied online problems in which the set of request sequences is drawn from a finite set, in search problems such as linear search, the request can be a hider's position in the environment, and thus the corresponding set is not only infinite, but also uncountable. This poses certain challenges when one defines bijective analysis in this context. To this end, we will adapt the concept of bijective ratio as follows:
% Given a search strategy $S$, and $l \in {\mathbb R}^+$, we denote by $D(S,l)$ the cost incurred by the searcher at the first time a total length equal to $l$ has been {\em discovered} (in both branches). 

% \begin{remark}
% For every strategy $S$, $\br(S)\leq \comp(S)$.
% \end{remark}
% \begin{proof}
% Given a hider position $H$, observe that $c(S,H) \geq D(|H|)$, since in order to locate $H$, the searcher must explore a length at least $|H|$. 
% Therefore,
% \[
% \comp(S)\sup_H \frac{c(S,H)}{|H|}\geq \sup_H \frac{D(S,|H|)}{|H|}=\sup_{l \in \mathbb{R}^+} \frac{D(S,l)}{l}=\br(S).
% \]
% \end{proof}

\paragraph*{Contribution}
We begin, in Section~\ref{sec:bijective.all}, by identifying the optimal tradeoff between the competitive ratio of a strategy and its discovery ratio (against all possible strategies). The result implies that there are strategies of discovery ratio $2+\epsilon$, for arbitrarily small $\epsilon>0$, which is tight. As corollary, we obtain that strategy {\tt doubling} has discovery ratio equal to 3. These results allow us to 
set up the framework and provide some intuition for our main results, but also demonstrate that the discovery ratio, on itself, does not lead to a useful classification of strategies, when one considers the entire space of strategies. 

Our main technical results are obtained in Section~\ref{sec:bijective.comp}. Here, we apply synthetically both the competitive and the discovery ratios. More precisely, we restrict our interest to the set of competitively optimal strategies, which we further analyze using the discovery ratio as a supplementary measure. We prove that the strategy 
{\tt aggressive}, which explores the branches to the furthest possible extent while satisfying the competitiveness constraint, has discovery ratio 1.6; moreover, we show that this is the optimal discovery ratio in this setting. 
In contrast, we show that the strategy {\tt doubling} has discovery ratio 2.333. In addition, we provide evidence that such ``aggressiveness'' is requisite. More precisely, we show that any competitively optimal strategy that is also optimal with respect to the discovery ratio must have the exact same behavior as the aggressive strategy in the first five iterations.  

% Finally, in Section~\ref{sec:motivation} we argue that the discovery ratio can be interpreted as the bijective ratio in a setting in which requests are not drawn from a discrete, finite set, but rather from an infinite, uncountable one. This interpretation can be useful for other online problems which are defined over a continuous setting of requests (e.g., $k$-server problems defined over a metric space rather than over a finite graph).

In terms of techniques, the main technical difficulty in establishing the discovery ratios stems from answering the following question: given a 
length $l \in {\mathbb R}^+$, what is the strategy $S$ that minimizes $D(S,l)$, and how can one express this minimum discovery cost? 
This is a type of {\em inverse} or {\em dual} problem that can be of independent interest in the context of search problems, in the spirit of a concept such as the {\em reach} of a strategy~\cite{HIKL99:fixed.distance}, also called {\em extent} in~\cite{jaillet:online} (and which is very useful in the competitive analysis of search strategies). We model this problem as a linear program for whose objective value we first give a lower bound; then we show this bound is tight by providing an explicit 9-competitive strategy which minimizes $D(S,l)$.

\paragraph*{Related work}
The linear search problem was first introduced and studied in works by Bellman~\cite{bellman} and Beck~\cite{beck:ls}. The generalization of linear search to $m$ concurrent, semi-infinite branches is known as {\em star search} or {\em ray search}; thus linear search is equivalent to star search for $m=2$. 
Optimal strategies for linear search under the (deterministic) competitive
ratio were first given by \cite{beck:yet.more}. Moreover~\cite{gal:minimax} gave optimal strategies for the generalized problem of star search, a result that was rediscovered later~\cite{yates:plane}.
Some of the related work includes the study of randomization~\cite{ray:2randomized}; multi-searcher strategies~\cite{alex:robots};
multi-target searching~\cite{hyperbolic,oil}; 
searching with turn cost~\cite{demaine:turn, Angelopoulos2017}; searching with an upper bound on the target distance~\cite{HIKL99:fixed.distance,revisiting:esa}; fault-tolerant search~\cite{czyzowicz_et_al:LIPIcs:2016:6797}; and the variant in which some probabilistic information on target placement is known~\cite{jaillet:online, informed.cows}. This list is not exclusive; see also Chapter 8 in the book~\cite{searchgames}.

Linear search and its generalization can model settings in which we seek an intelligent allocation of resources to tasks under uncertainty. For this reason, the problem and its solution often arises in the context of diverse fields such as AI (e.g., in the design of {\em interruptible algorithms}~\cite{steins}) and databases (e.g., {\em pipeline filter ordering}~\cite{Condon:2009:ADA:1497290.1497300}).

Strategy {\tt aggressive} has been studied in~\cite{HIKL99:fixed.distance,jaillet:online} in the special case of maximizing the {\em reach} of a strategy (which informally is the maximum possible extent to which the branches can be searched without violating competitiveness) when we do not know the distance of the target from the origin. Although this gives some intuition that {\tt aggressive} is indeed a good strategy, to the best of our knowledge, our work is the first that quantifies this intuition, in terms of comparing to other competitively optimal strategies using a well-defined performance measure.

\paragraph*{Preliminaries}
In the context of linear search, the searcher's strategy can be described as an (infinite) sequence of lengths at which the two branches (numbered 0,1, respectively) are searched. Formally, a search strategy is determined by an infinite sequence of {\em search segments} $\{x_0, x_1, \ldots\}$ 
such that $x_i>0$ and $x_{i+2}>x_i$ for all $i\in {\mathbb N}$, in the sense that in iteration $i$, the searcher starts from the origin, searches branch $i \bmod 2$ to distance $x_i$ from the origin, and then returns back to $O$. We require that the search segments induce a complete exploration of both branches of the line, in that for every $d \in {\mathbb R}^+$, there exist $i,j \in {\mathbb N}$ such that $x_{2i} \geq d$, and $x_{2j+1} \geq d$.

The constraint $x_{i+2}>x_i$ implies that the searcher explores a new portion of the line in each iteration. It is easy to see that any other strategy $X$ that does not conform to the above (namely, a strategy such that iterations $i$, $i+1$ search the same branch, or a strategy in which $x_{i+2} \leq x_i$ can be transformed to a conforming strategy $X'$ such that for any hider $H$, $c(X',S)\leq c(X,H)$). For convenience of notation, we will define $x_i$ to be equal to 0, for all $i<0.$ 
Given a strategy $X$, we define $T_n(X)$ (or simply $T_n$, when $X$ is clear from context) 
to be equal to $\sum_{i=0}^n x_i$. For $n<0$, we define $T_n:=0$.

We say that the searcher {\em turns} in iteration $i$ at the moment it switches directions during iteration $i$, namely when it completes the exploration of length $x_i$ and returns back to the origin. Moreover, at any given point in time $t$ (assuming a searcher of unit speed), the number of turns incurred by time $t$ is defined accordingly.

We will denote by $\Sigma$ the set of all search strategies, and by $\Sigma_c$ the subset of $\Sigma$ that consists of strategies with competitive ratio $c$. Thus $\Sigma_9$ is the set of competitively optimal strategies, and $\Sigma_\infty \equiv \Sigma$. When evaluating the competitive ratio, 
we will make the standard assumption that the target must be at distance at least 1 from $O$, since no strategy can have bounded competitive ratio if this distance can be arbitrarily small.

% Definition~\ref{def:br} shares one additional property with the definition of the bijective ratio over finite request spaces. Namely, we can argue that if each request (i.e, hider position) is equally likely, then the bijective ratio establishes a {\em stochastic dominance} relation between the costs incurred by the two compared strategies (see~\cite{bijective.ratio} for the relation between these concepts). To formalize this property, however, one needs to define a probability measure over the requests, which in turn necessitates that the line is finite. In Section~\ref{sec:discussion}, we argue that the upper bounds on the bijective ratio that we obtain for the infinite line (Sections~\ref{sec:bijective.all} and~\ref{sec:bijective.all}) caryy over to the finite line, and discuss further their significance in the context of stochastic dominance.  

\section{Strategies of optimal discovery ratio in $\Sigma$}
\label{sec:bijective.all}

We begin, by establishing the optimal tradeoff between the competitive ratio and the discovery ratio against all possible 
strategies. This will allow us to obtain strategies of optimal discovery ratio, and also setup some properties of the measure that will be useful
in Section~\ref{sec:bijective.comp}. 

Let $X,Y$, denote two strategies in $\Sigma$, with $X=(x_0,x_1,\ldots)$. From the definition of the discovery ratio we have that
\[
\br(X,Y)=\sup_{i \in \mathbb N} \sup_{\delta \in (0,x_{i}-x_{i-2}]} \frac{D(X,x_{i-1}+x_{i-2}+\delta)}{D(Y,x_{i-1}+x_{i-2}+\delta)}.
\]
Note that for $i = 0$, we have 
$
\frac{D(X,x_{i-1}+x_{i-2}+\delta)}{D(Y,x_{i-1}+x_{i-2}+\delta)} = \frac{D(X,\delta)}{D(Y,\delta)} \leq \frac{\delta}{\delta} = 1.
$
This is because for all $\delta \leq x_0$, $D(X,\delta) = \delta$, and for all $\delta>0$, $D(Y,\delta) \geq \delta$. 
Therefore, 
\begin{equation}
\br(X,Y)=\sup_{i \in \mathbb N^*} \sup_{\delta \in (0,x_{i}-x_{i-2}]} \frac{D(X,x_{i-1}+x_{i-2}+\delta)}{D(Y,x_{i-1}+x_{i-2}+\delta)}.
\label{eq:br.all.1}
\end{equation}
The following theorem provides an expression of the discovery ratio in terms of the search segments of the strategy.

\begin{theorem}
Let $X=(x_0,x_1, \ldots)$. Then  
\[
\br(X,\Sigma) = \sup_{i \in \mathbb N^* } \frac{2\sum_{j=0}^{i-1} x_j+x_{i-2}}{x_{i-1}+x_{i-2}}.
\]
%\shendan{For $i=0$, the fraction is not well defined.}
\label{thm:br.all.strategies}
\end{theorem}
\begin{proof}
Fix $Y \in \Sigma$. 
From the definition of search segments in $X$, we have that 
\begin{equation}
D(X,x_{i-1}+x_{i-2}+\delta)=2\sum_{j=0}^{i-1} x_j+x_{i-2}+\delta, \quad \textrm{for $\delta \in (0,x_{i}-x_{i-2}]$}.
\label{eq:searched.algo}
\end{equation}
Moreover, for every $Y$, we have 
\begin{equation}
D(Y,x_{i-1}+x_{i-2}+\delta) \geq x_{i-1}+x_{i-2}+\delta. 
\label{eq:searched.opt}
\end{equation}
Substituting~\eqref{eq:searched.algo} and~\eqref{eq:searched.opt} in~\eqref{eq:br.all.1} we obtain 
\begin{equation}
\br(X,Y) \leq \sup_{i \in \mathbb N^*} \sup_{\delta \in (0,x_{i}-x_{i-2}]} \frac{2\sum_{j=0}^{i-1} x_j+x_{i-2}+\delta}{x_{i-1}+x_{i-2}+\delta} \leq 
\sup_{i \in \mathbb N^*} \frac{2\sum_{j=0}^{i-1} x_j+x_{i-2}}{x_{i-1}+x_{i-2}}.
\label{eq:br.all.2}
\end{equation}
For the lower bound, consider a strategy $Y_i=(y_0^i,y_1^i, \ldots)$, for which $y_0^i=x_{i-1}+x_{i-2}+\delta$ (the values of $y_j^i$ for 
$j\neq 0$ are not significant, as long as $Y_i$ is a valid strategy). Clearly, 
$D(Y_i,x_{i-1}+x_{i-2}+\delta)=x_{i-1}+x_{i-2}+\delta$. Therefore,~\eqref{eq:br.all.1} implies
\begin{equation}
\br(X,Y_i) \geq \sup_{\delta \in (0,x_{i}-x_{i-2}]} \frac{2\sum_{j=0}^{i-1} x_j+x_{i-2}+\delta}{x_{i-1}+x_{i-2}+\delta}=
\frac{2\sum_{j=0}^{i-1} x_j+x_{i-2}}{x_{i-1}+x_{i-2}}.
\label{eq:bijective.Y_i}
\end{equation}
The lower bound on $\br(X,\Sigma)$ follows from $\br(X,\Sigma)\geq \sup_{i \in \mathbb N^*} \br(X,Y_i)$.
\end{proof}
In particular, note that for $i=2$, Theorem~\ref{thm:br.all.strategies} shows that for any strategy $X$, 
\[
\br(X,\Sigma)\geq \frac{3x_0+2x_1}{x_0+x_1}\geq 2.
\]
We will show that there exist strategies with discovery ratio arbitrarily close to 2, thus optimal for $\Sigma$. 
To this end, we will consider the geometric
search strategy defined as $G_\alpha=(1,\alpha,\alpha^2, \ldots)$, with $\alpha>1$.

\begin{lemma}
For $G_\alpha$ defined as above, we have $\br(G_\alpha,\Sigma)=\frac{2 \alpha^2+ \alpha-1}{\alpha^2-1}$.
\label{lemma:br.exponential.all}
\end{lemma}
\begin{proof}
From Theorem~\ref{thm:br.all.strategies} we have 
\[
%\begin{align*}
\br(G_\alpha, \Sigma) = \sup_{i \in \mathbb N^*} \frac{2\sum_{j=0}^{i-1}\alpha^j + \alpha^{i-2} }{\alpha^{i-1} + \alpha^{i-2}} 
	  = \sup_{i \in \mathbb N^*} \frac{2(\frac{\alpha^{i}-1}{\alpha-1}) + \alpha^{i-2}}{\alpha^{i-1} + \alpha^{i-2} } 
	  = \sup_{i \in \mathbb N^*} \frac{2(\alpha^i-1) + \alpha^{i-1} - \alpha^{i-2} }{\alpha^{i} - \alpha^{i-2}}. 
%\end{align*}
\]
The derivative of the function $f(i):=\frac{2(\alpha^i-1) + \alpha^{i-1} - \alpha^{i-2} }{\alpha^{i} - \alpha^{i-2} }$ in $i$ is
\[
f'(i)=\frac{2 \alpha^{2-i} \log \alpha}{\alpha^2-1},
\]
which is positive. Thus, $\sup_{i \in {\mathbb N}^*} f(i)=\lim_{i \rightarrow \infty} f(i)$, which gives
\[
%\begin{align*}
\br(G_\alpha, \Sigma) = \lim_{i \rightarrow +\infty}f(i) 
				=  \lim_{i \rightarrow +\infty} \frac{2(\alpha^i-1) + \alpha^{i-1} - \alpha^{i-2} }{\alpha^{i} - \alpha^{i-2} } 
				= \frac{2 \alpha^2 + \alpha- 1}{\alpha^2 -1}.
\]
\end{proof}

In particular, Lemma~\ref{lemma:br.exponential.all} shows that the discovery ratio of $G_\alpha$ tends to 2, as $\alpha \rightarrow \infty$, hence $G_\alpha$
has asymptotically optimal discovery ratio. However, we can show a stronger result, namely that $G_\alpha$ achieves the optimal trade-off between the discovery ratio and the competitive ratio. This is established in the following theorem, whose proof is based on results 
by Gal~\cite{gal80:search-games} and Schuierer~\cite{schuierer:lower} that, informally, lower-bounds the supremum of an infinite sequence of functionals by the supremum of simple functionals of a certain geometric sequence. 
We also note that the competitive ratio of $G_\alpha$ is known to be equal to $1+2\frac{\alpha^2}{\alpha-1}$ (and is minimized for $\alpha=2$). 

\begin{theorem}
For every strategy $X \in \Sigma$, there exists $\alpha>1$ such that $\br(X, \Sigma) \geq \frac{2\alpha^2 + \alpha- 1}{\alpha^2 -1}$ 
and $\comp(X) \geq 1+2\frac{\alpha^2}{\alpha-1}$. 
\label{thm:tradeoff.geometric}
\end{theorem}  
\begin{proof}
Let $X=(x_0,x_1,\ldots)$ denote a strategy in $\Sigma$. From~\eqref{eq:bijective.Y_i} we know that 
\[
\br(X,\Sigma) \geq \sup_i F_i(X),
\]
where $F_i(X)$ is defined as the functional $\frac{2\sum_{j=0}^{i-1} x_j+x_{i-2}}{x_{i-1}+x_{i-2}}$. Moreover, the competitive ratio of $X$
can be lower-bounded by 
\[
\comp(X) \geq  \sup_i F'_i(X), \ \textrm{ where } F'_i(X)=1+2\frac{\sum_{j=0}^{i+1}x_j}{x_i}.
\]
This follows easily by considering a hider placed at distance $x_i+\epsilon$, with $\epsilon \rightarrow 0$, 
at the branch that is searched by $X$ in iteration $i$.

In order to prove the theorem, we will make use of a result by Gal~\cite{gal80:search-games} and Schuierer~\cite{schuierer:lower} 
which we state here in a simplified form. This result will allow us to lower bound the supremum of a sequence of functionals by the supremum of simple functionals of a geometric sequence. Given an infinite sequence 
$X=(x_0, x_1, \ldots)$, define $X^{+i}=(x_i,x_{i+1}, \ldots)$ as the suffix of the sequence $X$ starting at $x_i$. 
Recall that $G_\alpha = (1,\alpha,\alpha^2,\ldots)$ is defined to be the geometric sequence in $\alpha$.

\begin{theorem}[\cite{gal80:search-games,schuierer:lower}]\label{thm:limit}
  Let $X = (x_0,x_1,\ldots)$ be a sequence of positive numbers, $r$
  an integer, and $\alpha = \limsup_{n\rightarrow\infty} (x_n)^{1/n}$, for
$\alpha\in \mathbb{R} \cup
\{+\infty\}$.  
Let $F_i$, $i \geq 0$ be a sequence of functionals which satisfy the following properties:
  \begin{enumerate}
  \item $F_i(X)$ only depends on $x_0,x_1, \ldots ,x_{i+r}$,
    \item $F_i(X)$ is continuous for all $x_k>0$, with $0 \leq k
      \leq i + r$, 
    \item $F_i(\lambda X) = F_i(X)$, for all $\lambda > 0$, 
    \item $F_i(X+Y) \leq \max(F_i(X),F_i(Y))$, and
    \item $F_{i+1}(X) \geq F_i(X^{k+1})$, for all $k \geq
      1$,  
  \end{enumerate}
  then
  \[
     \sup_{0 \leq i < \infty} F_i(X)
     \geq  \sup_{0 \leq i < \infty} F_i(G_\alpha).
  \]
\end{theorem}

It is easy to see that both $F_i(X)$ and $F'_i(X)$ satisfy the conditions of Theorem~\ref{thm:limit} (this also follows from 
Example~7.3 in~\cite{searchgames}). Thus, there exists $\alpha$ defined as in the statement of Theorem~\ref{thm:limit} such that 
\begin{eqnarray}
\br(X,\Sigma) &\geq& \sup_i F_i(G_\alpha)=\frac{2\sum_{j=0}^{i-1}\alpha^j+\alpha^{i-2}}{\alpha^{i-1}+\alpha^{i-2}}, \ \textrm{ and} 
 \label{eqn:functionals1} \\
\comp(X,\Sigma) &\geq& \sup_i F'_i(G_\alpha)=1+2\frac{\sum_{j=0}^{i+1}\alpha^j}{\alpha^i}.
\label{eqn:functionals2}
\end{eqnarray}
It is easy to verify that if $\alpha= 1$, then $\br(X,\Sigma),\comp(X,\Sigma) =\infty$. We can thus assume that $\alpha >1$, and 
thus obtain from~\eqref{eqn:functionals1},~\eqref{eqn:functionals2}, after some manipulations, that
\begin{eqnarray}
\br(X,\Sigma) &\geq& \sup_i \frac{2(\alpha^2-\frac{1}{\alpha^{i-2}})+\alpha-1}{\alpha^2-1}= \frac{2\alpha^2 + \alpha- 1}{\alpha^2 -1}, \nonumber 
\ \textrm{ and} \\
\comp(X,\Sigma) &\geq& 1+\sup_i \frac{2\sum_{j=0}^{i+1}\alpha^j}{\alpha^i}= \sup_i 1+2\frac{\alpha^2-\frac{1}{\alpha^i}}{\alpha-1}=1+2\frac{\alpha^2}{\alpha-1}, \nonumber
\end{eqnarray}
which concludes the proof.
\end{proof}

Figure~\ref{fig:Ga_cr_br} depicts this tradeoff, as attained by the search strategy $G_\alpha$ (see also Lemma~\ref{lemma:br.exponential.all}).

Note, however, that although $G_\alpha$, with $\alpha \rightarrow \infty$ has optimal discovery ratio, its competitive ratio is unbounded. Furthermore, strategy
${\tt doubling} \equiv G_2$ has optimal competitive ratio equal to 9, whereas its discovery ratio is equal to 3. This motivates the topic of the 
next section.  
% \begin{figure}[htb!]
% \centering
% \begin{tikzpicture}
%    \begin{axis}[
%        axis lines=center, xmin=1, xmax=10,
%        ymax = 25,
%        ylabel=ratios,
%        xlabel=$a$, scale=0.7
%        ]˘
%        \addplot [domain=1:10,samples=250, thick, blue] {1+2*x^2/(x-1)};
%           % node [pos=1, above left] {$y=\sqrt{x}$}
%        \addplot [domain=1:10,samples=250, thick, red ] {(2*x^2+x-1)/(x^2-1)};
%           % node [pos=0.5, below right] {$y=x$}
%    \end{axis}
%    \end{tikzpicture}
% \caption{Plots of the competitive ratio (in blue) and the discovery ratio (in red) of strategy $G_\alpha$, as function of $a$.}
% \label{fig:exponential-lowerbound}
% \end{figure}

% \begin{figure}[htb!]
% \centerline{\includegraphics[width=10cm]{figure_cr_br.pdf}}
% \caption{Plots of the competitive ratio and the bijective ratio of strategy $G_\alpha$, as function of $a$.}
% \label{fig:ga_cr_br}
%  \end{figure}

\section{The discovery ratio of competitively optimal strategies}
\label{sec:bijective.comp}
In this section we focus on strategies in $\Sigma_9$, namely the set of competitively optimal strategies. 
For any strategy $X \in \Sigma_9$, it is known that there is an infinite set of linear inequalities that relate its search segments, as shown in the following lemma 
(see, e.g,~\cite{jaillet:online}). For completeness, we include a proof of this lemma.  
\begin{lemma}
\label{lem:CR_9_inequalities}
The strategy $X=(x_0, x_1, x_2, \ldots)$ is in $\Sigma_9$ if and only if its segments satisfy the following inequalities
\[
1 \leq x_0 \leq 4, \quad x_1 \geq 1 \quad \textrm{and } \quad
x_n \leq 3x_{n-1} - \sum_{i=0}^{n-2}x_i, \quad \textrm{for all $n\geq 1$}.
\] 
\end{lemma}
\begin{proof}
It is well-known that the competitive ratio of $X$ is determined by seeking, for all $n \geq 0$, a target that is placed at distances $x_n+\epsilon$, where $\epsilon \rightarrow 0$, and in the same branch that is searched by $X$ in iteration $n$, namely branch $n \bmod 2$; call this target the {\em $n$-th target}. The cost incurred by $X$ for locating the $(n-1)$-th target, where $n \geq 1$ is equal to
$2(\sum_{i=0}^{n-1} x_i)+x_{n}+x_{n-1}+\epsilon$, whereas the optimal cost is $x_{n-1}+\epsilon$. From the definition of the competitive ratio, and since 
$\epsilon \rightarrow 0$, we obtain that 
\[
2\sum_{i=0}^{n-1} x_i+x_{n}+x_{n-1}\leq 9\cdot x_{n-1} \Rightarrow x_n \leq 3x_{n-1}-\sum_{i=0}^{n-2} x_i.
\] 
Moreover, we can obtain one more inequality that involves $x_0$, by assuming a target placed at distance 1 from $O$ in branch 1. 
Thus, we obtain that $2x_0+1\leq 9$, or, equivalently, $x_0\leq 4$.

Last, note that $x_0,x_1 \geq 1$ from the assumption that the target is at distance at least 1 from the origin. 
\end{proof}

We now define a class of strategies in $\Sigma_9$ as follows. For given $t\in[1,4]$, let $R_t$ denote the strategy whose search segments are determined 
by the linear recurrence
\[
x_0=t, \quad \textrm{and } \quad x_n = 3x_{n-1} - \sum_{i=0}^{n-2}x_i, \quad \textrm{for all $n\geq 1$}.
\]
In words, $R_t$ is such that for every $n>1$, the inequality relating $x_0, \ldots ,x_n$ is tight. The following lemma determines the search lengths of $R_t$ as function of $t$. The lemma also implies that $R_t$ is indeed a valid search strategy, for all $t\in[1,4]$, in that $x_n>x_{n-2}$, for all $n$, and $x_n \rightarrow \infty$, as $n\rightarrow \infty$. 
% Remark~\ref{lem:CR_9_inequalities} shows an upper bound on each attempt $x_i$ of a competitively optimal strategy $X(x_0, x_1, \ldots)$, for $i \geq 0$. We consider a class of strategies \textsc{reach} where each of its attempts $x_i$ reaches its upper bound described as in Remark~\ref{lem:CR_9_inequalities} except the first attempt $x_0$. Thus, any strategy in the class \textsc{reach} can be totally determined once we fix the first attempt. We can then identify the strategies in the class \textsc{reach} by its first attempt. Denote by $R_t$, the strategy in the class \textsc{Reach} with its first attempt equal to $t$. Then $R_t$ can be described as shown in the following lemma.
\begin{lemma}
\label{lem:reach_strategy}
The strategy $R_t$ is defined by the sequence $x_n=t(1+\frac{n}{2})2^n$, for $n \geq 0$. Moreover, 
$T_n(R_t) = t(n+1)2^n$.
\end{lemma}

\begin{proof}
The lemma is clearly true for $n\in\{0,1\}$. 
For $n \geq 2$, the equality $x_n=3x_{n-1}-\sum_{i=0}^{n-2}x_i$ implies that $T_n=\sum_{i=0}^nx_i=4x_{n-1}$. Therefore,
\[
T_n - T_{n-1} = 4x_{n-1} - 4x_{n-2} \Rightarrow x_n=4(x_{n-1}-x_{n-2}).
\] 
The characteristic polynomial of the above linear recurrence is $\xi^2 - 4\xi + 4$, with the unique root $\xi=2$.  
Thus, $x_n$ is of the form $x_n = (a+bn)2^n$, for $n \geq 0$, where $a$ and $b$ are determined by the initial conditions $x_0 = t$ and $x_1=3t$.
Summarizing, we obtain that for $n \geq 0$ we have that 
$
x_n=t(1+\frac{n}{2})2^n, \ \textrm{and } T_n = 4x_{n-1} = t(n+1)2^n.
$
\end{proof}
Among all strategies in $R_t$ we are interested, in particular, in the strategy $R_4$. This strategy has some intuitively appealing properties:
It maximizes the search segments in each iteration (see Lemma~\ref{lem:R_4_max}) and minimizes the number of turns required to discover a
certain length (as will be shown in Corollary~\ref{cor:min_turning}). Using the notation of the introduction, we can say that $R_4 \equiv {\tt aggressive}$. In this section we will show that $R_4$ has optimal discovery ratio among all competitively optimal strategies. Let us denote by $\bar x_i$ the search segment in the $i$-th iteration in $R_4$.
\begin{corollary}
\label{cor:R4_xn_and_Tn}
The strategy $R_4$ can be described by the sequence $\bar x_n = (n+2)2^{n+1}$, for $n \geq 0$. Moreover, $T_n(R_4) = (n+1)2^{n+2}$, for $n \geq 0$.
\end{corollary}
The following lemma shows that, for any given $n$, the total length discovered by any competitively optimal strategy $X$ at the turning point of the $n$-th iteration cannot exceed the corresponding length of $R_4$. Its proof can also be found in~\cite{jaillet:online}, but we give a different proof using ideas that we will apply later (Lemma~\ref{lem:optimal_cost}).
\begin{lemma}
\label{lem:R_4_max}
For every strategy $X=(x_0, x_1, \ldots)$ with $X \in \Sigma_9$, it holds that $x_n\leq \bar x_n$, for all $n \in \mathbb{N}$, 
where $\bar x_n$ is the search segment in the $n$-th iteration of $R_4$. Hence, in particular, we have $x_n+x_{n-1}\leq \bar x_n+\bar x_{n-1}$, for all $n \in \mathbb{N}$.
\end{lemma}
\begin{proof}
For a given $n \geq 0$, let $P_n$ denote the following linear program. 
\begin{align*}
\text{max} \quad &x_n \\
\text{subject to} \quad 
&1\leq x_0 \leq 4,\\
& x_1 \geq 1, \\
&x_i \leq 3x_{i-1} - \sum_{j=0}^{i-2}x_j,\quad  1\leq i \leq n.
\end{align*}
We will show, by induction on $i$, that for all $i\leq n$, 
\[
x_n \leq (i+2)2^{i-1}x_{n-i} - i2^{i-1}T_{n-i-1}(X).
\]
The lemma will then follow, since for $i = n$ we have  
\[
x_n \leq (n+2)2^{n-1}x_0 
  \leq (n+2)2^{n-1}\cdot4
  = (n+2)2^{n+1}
  = \bar x_n,
 \] 
where the last equality is due to Corollary~\ref{cor:R4_xn_and_Tn}.
We will now prove the claim. Note that, the base case, namely $i = 1$, follows directly from the LP constraint. For the induction hypothesis, suppose that for $i \geq 1$, it holds that 
\begin{equation}
x_n \leq (i+2)2^{i-1}x_{n-i} - i2^{i-1}T_{n-i-1}(X).
\label{eq:obj.induction}
\end{equation}
We will show that the claim holds for $i+1$. 
Since 
\begin{equation}
\label{eq:linear.ineq}
x_{n-i} \leq 3x_{n-i-1} - T_{n-i-2}(X),
\end{equation} then
\begin{align*}
		x_n  &\leq (i+2)2^{i-1}(3x_{n-i-1} - T_{n-i-2}(X)) - i2^{i-1}T_{n-i-1}(X) \tag{subst.~\eqref{eq:linear.ineq} in~\eqref{eq:obj.induction}} \\
		  &= (i+2)2^{i-1}(3x_{n-i-1} - T_{n-i-2}(X)) - i2^{i-1}(T_{n-i-2}(X) + x_{n-i-1}) \tag{def. $T_{n-i-1}$}\\
		  &= (i+3)2^ix_{n-i-1} + (i+1)2^iT_{n-i-2}(X), \tag{arranging terms}
\end{align*}
which completes the proof of the claim.
\end{proof}

Given strategy $X$ and $l \in \mathbb{R}^+$, define $m(X,l)$ as the number of turns that $X$ has performed by the time it discovers a total 
length equal to $l$. Also define
\[
m^*(l)=\inf_{X \in \Sigma_9} m(X,l),
\]
that is, $m^*(l)$ is the minimum number of turns that a competitively optimal strategy is required to perform in order to discover length equal to $l$. From the constraint $x_0 \leq 4$, it follows that clearly $m^*(l) = 0$, for $l\leq 4$. 
The following corollary to Lemma~\ref{lem:R_4_max} gives an expression for $m^*(l)$, for general values of $l$.  
\begin{corollary}
For given $l>4$, 
$
m^*(l)=m(R_4,l)=\min \{n \in \mathbb{N}_{\geq 1} : (3n+5)2^n \geq l\}.
$
\label{cor:min_turning}
\end{corollary}
\begin{proof}
From Lemma~\ref{lem:R_4_max}, the total length discovered by any $X \in \Sigma_9$ at the turning point of the $n$-th iteration cannot exceed $\bar x_n +\bar x_{n-1}$ for $n \geq 1$, which implies that $m^*(l) = n$, if $l \in (\bar x_{n-1} + \bar x_{n-2}, \bar x_{n} + \bar x_{n-1}]$ for $n \geq 1$. In other words,
\[
m^*(l) = \min \{n \in \mathbb{N}_{\geq 1} : \bar x_n + \bar x_{n-1} \geq l\}.
\]
From Corollary~\ref{cor:R4_xn_and_Tn}, we have $\bar x_n = (n+2)2^{n+1}$, for $n \geq 0$. Hence,
\[
m^*(l) = \min \{n \in \mathbb{N}_{\geq 1} : (3n+5)2^n \geq l\}.
\]
\end{proof}

The following lemma is a central technical result that is instrumental in establishing the bounds on the discovery ratio. 
For a given $l \in \mathbb{R}^+$, define 
\[
d^*(l)=\inf_{X \in \Sigma_9} D(X,l). 
\]
In words, $d^*(l)$ is the minimum cost at which a competitively optimal strategy can discover a length equal to $l$.
Trivially, $d^*(l)=l$ if $l\leq4$. Lemma~\ref{lem:optimal_cost} gives an expression of $d^*(l)$ for $l>4$ in terms of $m^*(l)$;
it also shows that there exists a $t \in (1,4]$ such that the strategy $R_t$ attains this minimum cost. 

We first give some motivation behind the purpose of the lemma. When considering general strategies in $\Sigma$, we used a lower
bound on the cost for discovering a length $l$ as given by~\eqref{eq:searched.opt}, and which corresponds to a strategy that 
never turns. However, this lower bound is very weak when 
one considers strategies in $\Sigma_9$. This is because a competitive strategy needs to turn sufficiently often, which affects considerably
the discovery costs. 

We also give some intuition about the proof. We show how to model the question by means of a linear program. Using the constraints of the LP,
we first obtain a lower bound on its objective in terms of the parameters $l$ and $m^*(l)$. In this process, we also obtain a lower bound on
the first segment of the strategy ($x_0$); this is denoted by $t$ in the proof. In the next step, we show that the strategy $R_t$ has discovery cost that matches the lower bound on the objective, which suffices to prove the result.
\begin{lemma}
\label{lem:optimal_cost}
For $l > 4$, it holds 
\[
d^*(l)= D(R_t,l) = l\cdot \frac{6m^*(l)+4}{3m^*(l)+5}, \quad 
\textrm{ where $t = l\cdot \frac{2^{2-m^*(l)}}{3m^*(l)+5} \in (1,4]$}.
\]
\end{lemma}

\begin{proof}
Let $X=(x_0, x_1, \ldots) \in \Sigma_9$ denote the strategy which minimizes the quantity $D(X,l)$. Then there must exist a smallest $n \geq m^*(l)$ such that the searcher discovers a total length $l$ during the $n$-th iteration. More precisely, suppose that this happens when the searcher is at branch $n \bmod 2$, and at some position $p$ (i.e., distance from $O$), with $p\in(x_{n-2}, x_n]$. Then we have $x_{n-1} + p = l$, and
\[
d^*(l) = D(X, l) =  2\sum_{i=0}^{n-1}x_i + p = 2\sum_{i=0}^{n-1}x_i + (l-x_{n-1}) = 2\sum_{i=0}^{n-2}x_i + x_{n-1} +l.
\]
Therefore, $d^*(l)$ is the objective of the following linear program.
% Without loss of generality, we can always assume that there exists some integer $n$ such that $x_n + x_{n-1} = A$. 
% Then $d^*(l)=D(X,l)$ is the objective of the following linear program:
\begin{align*}
\text{min} \quad &2\sum_{i=0}^{n-2}x_i + x_{n-1} + l&\\
\text{subject to} \quad 
&x_n + x_{n-1} \geq l,&\\
&1 \leq x_0 \leq 4,&\\
& x_{i-2} \leq x_i, &\quad i \in [2, n]\\
&1 \leq x_i \leq 3x_{i-1} - \sum_{j=0}^{i-2} x_j, &\quad i \in [1,n].
\end{align*}
Recall that $n \geq m^*(l)$ is a fixed integer.
Let Obj denote the objective value of the linear program. We claim that, for $1 \leq i \leq n$,
\[
x_{n-i} \geq \frac{2^{2-i}}{3i+5}l + \frac{3i-1}{3i+5}T_{n-i-1} \quad \textrm{and} \quad
\textrm{Obj} \geq \frac{6i+4}{3i+5}l + \frac{9\cdot2^i}{3i+5}T_{n-i-1}.
\]
The claim provides a lower bound of the objective, since for $i=n$ it implies that 
\[
x_0 \geq \frac{2^{2-n}}{3n+5}l \quad \textrm{and} \quad
\textrm{Obj} \geq \frac{6n+4}{3n+5}l \geq \frac{6m^*(l)+4}{3m^*(l)+5}l,
\]
where the last inequality follows from the fact $n \geq m^*(l)$. We will argue later that this lower bound is tight, 

First, we prove the claim, by induction on $i$, for all $i\leq n$. We first show the base case, namely $i=1$.  
Since $x_n \leq 3x_{n-1} - T_{n-2}$ and $x_n + x_{n-1} \geq l$, it follows that  
\[
x_{n-1} \geq l - x_n  \geq l - (3x_{n-1} - T_{n-2}) \Rightarrow
x_{n-1} \geq \frac{l}{4} + \frac{T_{n-2}}{4}, \quad \textrm{hence}
\]
\[
\textrm{Obj} = l + 2T_{n-2} + x_{n-1}
	\geq l + 2T_{n-2} + \frac{l}{4} + \frac{T_{n-2}}{4}
	= \frac{5}{4}l + \frac{9}{4}T_{n-2},
\]
thus the base case holds. For the induction step,  suppose that 
\[
x_{n-i} \geq \frac{2^{2-i}}{3i+5}l + \frac{3i-1}{3i+5}T_{n-i-1} \quad \textrm{and} \quad 
\textrm{Obj} \geq \frac{6i+4}{3i+5}l + \frac{9\cdot2^i}{3i+5}T_{n-i-1}.
\]
Then,
\begin{align*}
3x_{n-i-1}-T_{n-i-2} &\geq x_{n-i} \tag{by LP constraint} \\
					 &\geq \frac{2^{2-i}}{3i+5}l + \frac{3i-1}{3i+5}T_{n-i-1} \tag{ind. hyp.}\\
					 &= \frac{2^{2-i}}{3i+5}l + \frac{3i-1}{3i+5}(T_{n-i-2} + x_{n-i-1}) \tag{def. $T_{n-i-1}$}
\end{align*}
By rearranging terms in the above inequality we obtain
\begin{align*}
(3-\frac{3i-1}{3i+5})x_{n-i-1} &\geq \frac{2^{2-i}}{3i+5}l  + (1+\frac{3i-1}{3i+5})T_{n-i-2} \Rightarrow \\
\frac{6i+16}{3i+5}x_{n-i-1} &\geq  \frac{2^{2-i}}{3i+5}l + \frac{6i+4}{3i+5}T_{n-i-2}\Rightarrow 
x_{n-i-1} \geq \frac{2^{1-i}}{3i+8}l + \frac{3i+2}{3i+8}T_{n-i-2},
\end{align*}
and 
\begin{align*}
\textrm{Obj} &\geq \frac{6i+4}{3i+5}l + \frac{9\cdot2^i}{3i+5}T_{n-i-1} \tag{ind. hyp.}\\
	&= \frac{6i+4}{3i+5}l + \frac{9\cdot2^i}{3i+5}(T_{n-i-2} + x_{n-i-1}) \tag{def. $T_{n-i-1}$}\\
	&\geq \frac{6i+4}{3i+5}l + \frac{9\cdot2^i}{3i+5}T_{n-i-2} + \frac{9\cdot2^i}{3i+5}(\frac{2^{1-i}}{3i+8}l + \frac{3i+2}{3i+8}T_{n-i-2}) 
	\tag{ind. hyp.}\\
	&= \frac{6i+10}{3i+8}l + \frac{9\cdot2^{i+1}}{3i+8}T_{n-i-2}.
\end{align*}
This concludes the proof of the claim, which settles the lower bound on $d^*(l)$.  It remains to show that this bound is tight. Consider the strategy $R_t$, with $t = \frac{2^{2-m^*(l)}}{3m^*(l)+5}l$.
In what follows we will show that $R_t$ is a feasible solution of the LP, and that $D(R_t, l) = \frac{6m^*(l)+4}{3m^*(l)+5}l$. 

First, we show that $t\in (1,4]$. For the upper bound, from Corollary~\ref{cor:min_turning}, we have $(3m^*(l)+5)2^{m^*(l)} \geq l$, which implies that
\[
1 \geq l\cdot \frac{2^{-m^*(l)}}{3m^*(l)+5} \Rightarrow 4 \geq l\cdot \frac{2^{2-m^*(l)}}{3m^*(l)+5} \Rightarrow 4 \geq t.
\]
In order to show that $t>1$, consider first the case $l \in (4, 5]$. Then $m^*(l) = 1$, which implies that 
\[
t = \frac{2^{2-m^*(l)}}{3m^*(l)+5}l = \frac{l}{4} \geq 1.
\]
Moreover, if $l > 5$, by Corollary~\ref{cor:min_turning}, $m^*(l)$ is the smallest integer solution of the inequality $(3n+5)2^n \geq l$, then $(3m^*(l)+2)2^{m^*(l)-1} < l$, hence
\begin{align*}
t &= \frac{2^{2-m^*(l)}}{3m^*(l)+5}l = \frac{4l}{(3m^*(l)+5)2^{m^*(l)}} = \frac{2l}{(3m^*(l)+2)2^{m^*(l)-1} \cdot \frac{3m^*(l)+5}{3m^*(l)+2}}\\
  &> \frac{2l}{l \cdot \frac{3m^*(l)+5}{3m^*(l)+2}} = \frac{6m^*(l)+4}{3m^*(l)+5} > 1.
\end{align*}
The last inequality holds since we have $m^*(l) \geq 1$, for $l > 5$. This concludes that $t\in (1,4]$, and $R_t$ is a feasible solution of the LP since $R_t$ satisfies all other constraints by its definition. 

It remains thus to show that $D(R_t, l) = \frac{6m^*(l)+4}{3m^*(l)+5}l$. By Lemma~\ref{lem:reach_strategy}, we have 
\begin{align*}
x_{m^*(l)} + x_{m^*(l)-1} &= t\left(1+\frac{m^*(l)}{2}\right)2^{m^*(l)} + t\left(1+\frac{m^*(l)-1}{2}\right)2^{m^*(l)-1}\\
		&= t\cdot 2^{m^*(l)}\cdot \frac{3m^*(l)+ 5}{4} = \frac{2^{2-m^*(l)}}{3m^*(l)+5}l \cdot 2^{m^*(l)}\cdot \frac{3m^*(l)+ 5}{4} = l.
\end{align*}
Then $R_t$ has exactly discovered a total length $l$ right before the $m^*(l)$-th turn. Hence,
\begin{align*}
D(R_t, l) &= 2T_{m^*(l)-2} + x_{m^*(l)-1} + l \\
			 &= t\cdot \left(m^*(l)-1\right)2^{m^*(l)-1} + t\cdot \left(1+\frac{m^*(l)-1}{2}\right)2^{m^*(l)-1} + l \tag{by Lemma~\ref{lem:reach_strategy}}\\
			 &= t\cdot\frac{(3m^*(l)-1)2^{m^*(l)}}{4} + l \tag{arranging terms}\\
			 &= \frac{2^{2-m^*(l)}}{3m^*(l)+5}l \cdot\frac{(3m^*(l)-1)2^{m^*(l)}}{4} + l \tag{substituting $t$}\\
			 &= \left(\frac{3m^*(l)-1}{3m^*(l)+5} + 1\right)\cdot l = \frac{6m^*(l)+4}{3m^*(l)+5}\cdot l 			\tag{arranging terms}.
\end{align*}
This concludes the proof of the lemma.
\end{proof}

We are now ready to prove the main results of this section. Recall that for any two strategies $X,Y$, $\br(X,Y)$ is given by~\eqref{eq:br.all.1}.
Combining with~\eqref{eq:searched.algo}, as well as with the fact that for $Y \in \Sigma_9$, we have that $D(Y,l)\geq d^*(l)$, (from the definition 
of $d^*$), we obtain that 
\begin{equation}
\br(X,\Sigma_9) = \sup_{i \in \mathbb N^*} \sup_{\delta \in (0,x_{i}-x_{i-2}]} F_i(X,\delta), \quad \textrm{where } 
F_i(X,\delta)=\frac{2\sum_{j=0}^{i-1} x_j+x_{i-2}+\delta}{d^*(x_{i-1}+x_{i-2}+\delta)}.
\label{eq:general.9} 
\end{equation}

Recall that for the strategy $R_4=(\bar x_0, \bar x_1, \ldots)$, its segments $\bar x_i$ are given in Corollary~\ref{cor:R4_xn_and_Tn}.

\begin{theorem}
For the strategy $R_4$ it holds that $\br(R_4, \Sigma_9)=8/5$. 
\label{thm:R4.optimal}
\end{theorem}

\begin{proof}
We will express the discovery ratio using~\eqref{eq:general.9}. For $i=1$, and $\delta \in (0,\bar x_1]$, we have that 
\[
F_1(R_4, \delta) = \frac{2\bar x_0 + \delta}{d^*(\bar x_0 + \delta)}=
				 \frac{8+\delta}{d^*(4+\delta)}.
\]
From Lemma~\ref{lem:optimal_cost}, $d^*(4+\delta) = (4+\delta)\cdot\frac{6\cdot1+4}{3\cdot1+5} = \frac{5(4+\delta)}{4}$; this is because
$1 \leq m^*(4+\delta)\leq m^*(16)=1$. Then,
\begin{equation}
F_1(R_4, \delta) = \frac{8+\delta}{\frac{5(4+\delta)}{4}} = \frac{32+4\delta}{20+5\delta}, \ \textrm{hence } \sup_{\delta \in (0,\bar x_1]} F_1(R_4,\delta)=\frac{8}{5}.
\label{eq:R4:F1}
\end{equation}
For given $i \geq 2$, and $\delta \in (0,\bar x_{i}-\bar x_{i-2}]$, we have
\[
F_i(R_4, \delta) = \frac{2T_{i-1} + \bar x_{i-2} + \delta}{d^*(\bar x_{i-1}+\bar x_{i-2}+\delta)},
\]
where $T_{i-1}$ is given by Corollary~\ref{cor:R4_xn_and_Tn}. Moreover, from Lemma~\ref{lem:optimal_cost} we have that 
\[
d^*(\bar x_{i-1}+\bar x_{i-2}+\delta) =(\bar x_{i-1}+\bar x_{i-2}+\delta) \cdot \frac{6m^*(\bar x_{i-1}+\bar x_{i-2}+\delta)+4}{3m^*(\bar x_{i-1}+\bar x_{i-2}+\delta)+5}
= (\bar x_{i-1}+\bar x_{i-2}+\delta) \cdot \frac{6i+4}{3i+5},
\]
where the last equality follows from the fact that $m^*(\bar x_{i-1}+\bar x_{i-2}+\delta)=i$. This is because
\[
i \leq m^*(\bar x_{i-1}+\bar x_{i-2}+\delta)\leq m^*(\bar x_{i-1}+\bar x_{i-2}+\bar x_{i}-\bar x_{i-2})=m^*(\bar x_i+\bar x_{i-1})=i.
\]
Substituting with the values of the search segments as well as $T_{i-1}$, we obtain that
\[
%\begin{align*}
F_i(R_4, \delta) 
%&= \frac{2T_{i-1} + \bar x_{i-2} + \delta}{d^*(\bar x_{i-1}+\bar x_{i-2}+\delta)}\\
				= \frac{i\cdot2^{i+2} + i\cdot2^{i-1}+\delta}{((i+1)2^i+i\cdot2^{i-1}+\delta)\cdot \frac{6i+4}{3i+5}}
				= \frac{9i\cdot2^{i-1}+\delta}{((3i+2)2^{i-1}+\delta)\cdot \frac{6i+4}{3i+5}}.
%\end{align*}
\]
Since 
\[
\frac{\partial F_i(R_4, \delta)}{\partial \delta} = -\frac{2^{i+1}(3i-1)(3i+5)}{(3i+2)(2^n(3i+2)+2\delta)^2} \leq 0,
\]
then $F_i(R_4, \delta)$ is monotone decreasing in $\delta$. Thus
\[
\sup_{\delta \in (0,\bar x_{i}-\bar x_{i-2}]} F_i(R_4, \delta) = \frac{9i\cdot2^{i-1}}{((3i+2)2^{i-1})\cdot \frac{6i+4}{3i+5}}
													= \frac{9i(3i+5)}{(3i+2)(6i+4)},
\]
and then
\begin{equation}
\sup_{i \in \mathbb N_{i \geq 2}} \sup_{\delta \in (0,\bar x_{i}-\bar x_{i-2}]} F_i(R_4, \delta) = \frac{(9\cdot2)(3\cdot2+5)}{(3\cdot 2+2)(6\cdot 2+4)} 
= \frac{99}{64}< \frac{8}{5}.
\label{eq:R4:Fi}
\end{equation}
Combining~\eqref{eq:general.9},~\eqref{eq:R4:F1} and~\eqref{eq:R4:Fi} yields the proof of the theorem. 
\end{proof}

The following theorem shows that $R_4$ has optimal discovery ratio among all competitively optimal strategies. 
\begin{theorem}
For every strategy $X \in \Sigma_9$, we have $\br(X,\Sigma_9) \geq \frac{8}{5}$.
\label{thm:bij.lower}
\end{theorem}
\begin{proof}
Let $X=(x_0,\ldots)$. We will consider two cases, depending on whether $x_0<4$ or $x_0=4$. Suppose, first, that $x_0<4$. In this case, for
sufficiently small $\epsilon$, we have $m^*(x_0+\epsilon)=0$, which implies 
that $d^*(x_0+\epsilon)=x_0+\epsilon$, and therefore.
\[
F_1(X,\epsilon)=\frac{2x_0+\epsilon}{d^*(x_0+\epsilon)}=\frac{2x_0+\epsilon}{x_0+\epsilon},
\]
from which we obtain that 
\[
\sup_{\delta \in (0,x_1]} F_1(X,\delta)\geq F_1(X,\epsilon) \geq \frac{2x_0+\epsilon}{x_0+\epsilon} \rightarrow 2, \
\textrm{as $\epsilon \rightarrow 0^+$}.
\]
Next, suppose that $x_0=4$. In this case, for $\delta \in (0,x_1]$, it readily follows that $F_1(X,\delta)=F_1(R_4,\delta)$. 
Therefore, from~\eqref{eq:R4:F1}, we have that 
\[
\sup_{\delta \in (0,x_1]} F_1(X,\delta) = \sup_{\delta \in (0,x_1]} \frac{32+4\delta}{20+5\delta} =\frac{8}{5}.
\]
The lower bound follows directly from~\eqref{eq:general.9}.
\end{proof}

Recall that $G_2$ is the standard doubling strategy $G_2=(2^0,2^1,\ldots)$. The following theorem shows that within $\Sigma_9$, $G_2$ has worse discovery ratio than $R_4$. The proof follows along the lines of the  proof of Theorem~\ref{thm:R4.optimal}, where instead of using the search segments $\bar x_i$ of $R_4$, we use the search segment $x_i = 2^i$ of $G_2$.

\begin{theorem}
For the strategy $G_2=(x_0, x_1, \ldots)$,  we have $\br(G_2,\Sigma_9)=\frac{7}{3}.$
\label{thm:G2.bijective}
\end{theorem}
\begin{proof}
We will express the discovery ratio using~\eqref{eq:general.9}. 
%see Figure~\ref{fig:Ga_cr_br} for an illustration. 
For $i=1$, and $\delta \in (0,x_1]$, we have that 
\[
F_1(G_2, \delta) = \frac{2x_0 + \delta}{d^*(x_0 + \delta)}=
				 \frac{2+\delta}{d^*(1+\delta)}.
\]
From the definition, $d^*(1+\delta) = 1+\delta$; this is because
$0 \leq m^*(1+\delta)\leq m^*(3)=0$. Then,
\begin{equation}
F_1(G_2, \delta) = \frac{2+\delta}{1+\delta}, \ \textrm{hence } \sup_{\delta \in (0,\bar x_1]} F_1(G_2,\delta)=2.
\label{eq:G2:F1}
\end{equation}
For $i = 2$ and $\delta \in (0,x_2 - x_0]$, we have that 
\[
F_2(G_2, \delta) = \frac{3x_0 + 2x_1 + \delta}{d^*(x_0 + x_1 + \delta)}=
				 \frac{7+\delta}{d^*(3+\delta)}.
\]
From Lemma~\ref{lem:optimal_cost}, $d^*(3+\delta)$ either equals to $3+\delta$ if $\delta \in (0,1]$, or equals to $(3+\delta)\cdot\frac{6\cdot1+4}{3\cdot1+5} = \frac{5(3+\delta)}{4}$ if $\delta \in (1, x_2-x_0]$. This is because
$0 \leq m^*(3+\delta)\leq m^*(6)=1$. Then, for $\delta \in (0,1]$,
\begin{equation}
F_2(G_2, \delta) = \frac{7+\delta}{3+\delta}, \ \textrm{hence } \sup_{\delta \in (0,1]} F_2(G_2,\delta)=\frac{7}{3}.
\label{eq:R4:F2:1}
\end{equation}
For $\delta \in (1, x_2-x_0]$, 
\begin{equation}
F_2(G_2, \delta) = \frac{7+\delta}{\frac{5(3+\delta)}{4}}=\frac{28+4\delta}{15+5\delta}, \ \textrm{hence } \sup_{\delta \in (1, x_2 - x_0]} F_2(G_2,\delta)=\frac{28}{15}.
\label{eq:R4:F2:2}
\end{equation}
Combining~\eqref{eq:R4:F2:1} and~\eqref{eq:R4:F2:2} yields
\begin{equation}
\sup_{\delta \in (0,x_2-x_0]} F_2(G_2,\delta)=\frac{7}{3}.
\label{eq:G2:F2}
\end{equation}
For given $i \geq 3$, and $\delta \in (0,x_{i}-x_{i-2}]$, we have
\[
F_i(G_2, \delta) = \frac{2T_{i-1} + x_{i-2} + \delta}{d^*(x_{i-1}+x_{i-2}+\delta)} =\frac{2^{i+1}-2+2^{i-2}+\delta}{d^*(2^{i-1}+2^{i-2}+\delta)} = \frac{9\cdot2^{i-2}-2+\delta}{d^*(2^{i-1}+2^{i-2}+\delta)}.
\]
Moreover, from Lemma~\ref{lem:optimal_cost} we have that 
\[
d^*(2^{i-1}+2^{i-2}+\delta) =(2^{i-1}+2^{i-2}+\delta) \cdot \frac{6m^*(2^{i-1}+2^{i-2}+\delta)+4}{3m^*(2^{i-1}+2^{i-2}+\delta)+5}.
\]
and 
\[
m^*(x_{i-1}+x_{i-2}+\delta) \geq m^*(x_{i-1}+x_{i-2}) = m^*(3\cdot2^{i-2}).
\]
For $i \in \{3, 4\}$ and $\delta \in (0, x_i - x_{i-2}]$, $m^*(x_{i-1}+x_{i-2}+\delta) \geq  m^*(3\cdot2^{i-2}) \geq 1$,
then 
\[
d^*(2^{i-1}+2^{i-2}+\delta) \geq (2^{i-1}+2^{i-2}+\delta) \cdot \frac{6\cdot1+4}{3\cdot1+5} = \frac{15\cdot2^{i-2}+5\delta}{4},
\]
hence, for $i = 3$,
\[
F_3(G_2, \delta) = \frac{9\cdot2^{i-2}-2+\delta}{d^*(2^{i-1}+2^{i-2}+\delta)} \leq \frac{16+\delta}{\frac{30+5\delta}{4}}=\frac{64+4\delta}{30+5\delta}.
\]
We obtain that
\begin{equation}
\sup_{\delta \in (0,x_3 - x_1]} F_3(G_2,\delta) \leq \frac{64}{30}.
\label{eq:G2:F3}
\end{equation}
For $i = 4$,
\[
F_4(G_2, \delta) = \frac{9\cdot2^{i-2}-2+\delta}{d^*(2^{i-1}+2^{i-2}+\delta)} \leq \frac{32+\delta}{\frac{}{60+5\delta}{4}}=\frac{128+4\delta}{60+5\delta}.
\]
We obtain that
\begin{equation}
\sup_{\delta \in (0,x_4 - x_2]} F_4(G_2,\delta) \leq \frac{128}{60}.
\label{eq:G2:F4}
\end{equation}
For $i \geq 5$ and $\delta \in (0, x_i - x_{i-2}]$, $m^*(x_{i-1}+x_{i-2}+\delta) \geq m^*(3\cdot2^{i-2}) \geq 2$,
then 
\[
d^*(2^{i-1}+2^{i-2}+\delta) \geq (2^{i-1}+2^{i-2}+\delta) \cdot \frac{6\cdot2+4}{3\cdot2+5} = \frac{48\cdot2^{i-2}+16\delta}{11},
\]
and
\[
F_i(G_2, \delta) = \frac{9\cdot2^{i-2}-2+\delta}{d^*(2^{i-1}+2^{i-2}+\delta)} \leq \frac{9\cdot2^{i-2}-2+\delta}{\frac{48\cdot2^{i-2}+16\delta}{11}} = \frac{99\cdot2^{i-2}-22}{48\cdot2^{i-2}} \leq \frac{99}{48},
\]{}
hence, for $i \geq 5$,
\begin{equation}
\sup_{\delta \in (0,x_i - x_{i-2}]} F_i(G_2,\delta) \leq \frac{99}{48}.
\label{eq:G2:F5}
\end{equation}

Combining~\eqref{eq:G2:F1},~\eqref{eq:G2:F2},~\eqref{eq:G2:F3},~\eqref{eq:G2:F4} and~\eqref{eq:G2:F5} yields the proof of the theorem. 
\end{proof}

A natural question arises: Is $R_4$ the unique strategy of optimal discovery ratio in $\Sigma_9$? The following theorem provides 
evidence that optimal strategies cannot be radically different than $R_4$, in that they must mimic it in the first few iterations. 

\begin{theorem}
Strategy $X=(x_0, x_1, \ldots) \in \Sigma_9$, has optimal discovery ratio in $\Sigma_9$ only if 
$x_i=\bar x_i$, for $0\leq i \leq 4$. 
%$(x_0, x_1, x_2, x_3, x_4) = (\bar x_0, \bar x_1, \bar x_2, \bar x_3, \bar x_4)$.
\label{thm:bijectively.optimal.c9}
\end{theorem}

\begin{proof}
Consider a strategy $X(x_0, x_1, \ldots) \in \Sigma_9$. Recall that the discovery ratio of $X$ is given by Equation~\eqref{eq:general.9}.
% \[
% \br(X,\Sigma_9)=\sup_{i \in \mathbb N^*} \sup_{\delta \in (0,x_{i}-x_{i-2}]} F_{i}(X, \delta),
% \]
% where 
% \[
% F_{i}(X, \delta) = \frac{2\cdot\sum_{j=0}^{i-1}x_j + x_{i-2} + \delta}{d^*(x_{i-1}+x_{i-2}+\delta)}.
% \]
We will prove the theorem by induction on $i$.  

We first show the base case, namely $i = 0$. The base case holds by the argument used in the proof of Theorem~\ref{thm:bij.lower} which shows that if $x_0<4$, then $\br(X,\Sigma_9) \geq 2$. For the induction step, suppose that, if $X$ has optimal discovery ratio then for $j \in [0,i]$, $x_j = \bar x_j$, with  $i < 4$. 
%Then, we consider two cases, depending on whether $x_{i+1} < \bar x_{i+1}$ or $x_{i+1} = \bar x_{i+1}$. 
We will show $x_{i+1}=\bar x_{i+1}$ by contradiction, hence assume $x_{i+1} < \bar x_{i+1}$. 
%For the induction step, suppose, by way of contradiction, that $x_{i+1} < \bar x_{i+1}$. 
For sufficiently small $\epsilon > 0$, we have 
\begin{align*}
m^*(x_{i+1}+x_{i} + \epsilon) &= m^*(x_{i+1} + \bar x_i + \epsilon) \tag{by induction hypothesis}\\
& \leq m^*(\bar x_{i+1} + \bar x_i) \tag{by monotonicity of $m^*$ and Lemma~\ref{lem:R_4_max}}\\
&= i+1, \tag{by definition of $m^*$}
\end{align*}
which implies that, by Lemma~\ref{lem:optimal_cost}, 
\begin{equation}
\label{eq:d_last}
d^*(x_{i}+x_{i-1}+\epsilon) = (x_{i}+x_{i-1}+\epsilon)\cdot\frac{6\cdot m^*(x_{i+1}+x_{i} + \epsilon)+4}{3\cdot m^*(x_{i+1}+x_{i} + \epsilon)+5} \leq  (x_{i}+x_{i-1}+\epsilon)\cdot\frac{6\cdot(i+1)+4}{3\cdot(i+1)+5}.
\end{equation}
Therefore
\begin{align*}
F_{i+2}(X, \epsilon) &= \frac{2\cdot\sum_{j=0}^{i+1}x_j + x_{i} + \epsilon}{d^*(x_{i+1}+x_{i}+\epsilon)}\\
&= \frac{2T_{i}(R_4)+ 2x_{i+1} + \bar x_{i}+\epsilon}{d^*(x_{i+1}+\bar x_{i}+\epsilon)} \tag{by ind. hyp.}\\
	& \geq \frac{2T_{i}(R_4)+ 2x_{i+1} + \bar x_{i}+\epsilon}{(x_{i+1}+\bar x_{i}+\epsilon)\cdot\frac{6\cdot(i+1)+4}{3\cdot(i+1)+5}} \tag{Equation~\eqref{eq:d_last}}\\
	&= \frac{(i+1)2^{i+3} + (i+2)2^{i+1} + 2x_{i+1}+\epsilon}{(x_{i+1}+(i+2)2^{i+1}+\epsilon)\cdot\frac{6\cdot(i+1)+4}{3\cdot(i+1)+5}} \tag{Corollary~\ref{cor:R4_xn_and_Tn}}\\
	&\geq  \frac{(i+1)2^{i+3}+(i+2)2^{i+1}+(i+3)2^{i+3}+\epsilon}{(i+3)2^{i+2}+(i+2)2^{i+1}+\epsilon}\cdot\frac{3i+8}{6i+10}. \tag{monoton. on $x_{i+1}$}\\
\end{align*}
Hence
\[
\sup_{\delta \in (0,x_{i+2} - x_{i}]} F_{i+2}(X,\delta) \geq \frac{(i+1)2^{i+3}+(i+2)2^{i+1}+(i+3)2^{i+3}}{(i+3)2^{i+2}+(i+2)2^{i+1}}\cdot\frac{3i+8}{6i+10} = \frac{9i+18}{6i+10},
\]
which is greater than $\frac{8}{5}$ if $i \leq 3$. We conclude, from~\eqref{eq:general.9} that $\br(X,\Sigma_9)>8/5$, which is a contradiction. 
\end{proof}

We conclude with a simple observation.
\begin{lemma}
	We have $\textrm{dr}(R_t,\Sigma)=3$ for any $t>1$. 
\end{lemma}
\begin{proof}
	By Theorem~\ref{thm:br.all.strategies} we have
	\begin{align*}
	 	\textrm{dr}(R_t,\Sigma)
	 	&= 
	 	\sup_{i \in \mathbb N^* } \frac{2\sum_{j=0}^{i-1} x_j+x_{i-2}}{x_{i-1}+x_{i-2}} \\
	 	&= \sup_{i \in \mathbb N^* } \frac{2n2^{n-1} + (1+\frac{n-2}{2})2^{n-2}}{(1+\frac{n-1}{2})2^{n-1}+(1+\frac{n-2}{2})2^{n-2}} \\
	 	&= \sup_{i \in \mathbb N^* } \frac{n+\frac14+\frac{n-2}8}{\frac12+\frac{n-1}4+\frac14+\frac{n-2}8} \\
	 	&= \sup_{i \in \mathbb N^* } \frac{9n}{3n+2} \\
	 	& = 3.
	 \end{align*} 
\end{proof}

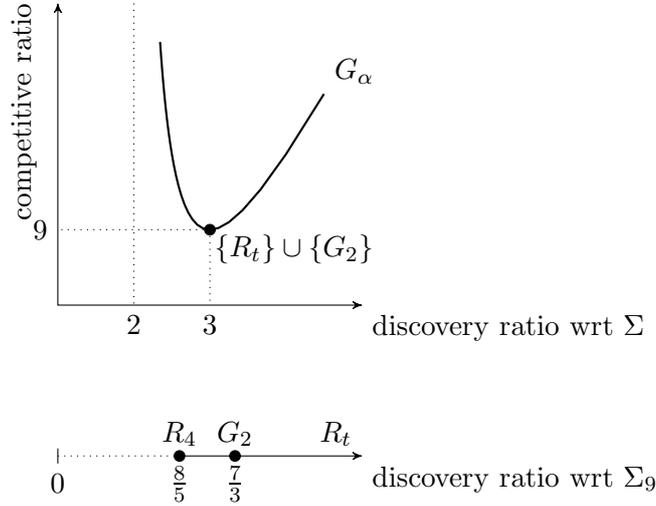
\begin{figure}[htb!]
\begin{center}
  \begin{tikzpicture}
    \coordinate (y) at (5,8);
    \coordinate (x) at (1,12);
    \draw[<->] (y) node[below right] {discovery ratio wrt $\Sigma$} -- (1,8) --  (x) node[above left=2mm, rotate=90]
    {competitive ratio};
    \path
    coordinate (start) at (6,0)
    coordinate (c1) at +(3,5)
    coordinate (c2) at +(1.75,5)
    coordinate (slut) at (.5,2.7)
    coordinate (top) at (2,4.2);

    \draw[important line]  (4.500000,10.800000) node[above right] {$G_\alpha$} -- (4.000000,10.000000) -- (3.666667,9.533333) -- (3.428571,9.257143) -- (3.250000,9.100000) -- (3.111111,9.022222) -- (3.000000,9.000000) -- (2.909091,9.018182) -- (2.833333,9.066667) -- (2.769231,9.138462) -- (2.714286,9.228571) -- (2.666667,9.333333) -- (2.625000,9.450000) -- (2.588235,9.576471) -- (2.555556,9.711111) -- (2.526316,9.852632) -- (2.500000,10.000000) -- (2.476190,10.152381) -- (2.454545,10.309091) -- (2.434783,10.469565) -- (2.416667,10.633333) -- (2.400000,10.800000) -- (2.384615,10.969231) -- (2.370370,11.140741) -- (2.357143,11.314286) -- (2.344828,11.489655);% node[right,rotate=90] {lim=2};     
    \draw[dotted] (3,9) -- (3,8) node[below] {3};
    \draw[dotted] (2,12) -- (2,8) node[below] {2};
    \draw[dotted] (3,9) -- (1,9) node[left] {9};
    \filldraw [black] 
     (3,9) circle (2pt) node[below right=-1mm] {$\{R_t\}\cup\{G_2\}$};

    \draw[->] (2.6,6)  -- (5,6) node[below right] {discovery ratio wrt $\Sigma_9$} node[above left] {$R_t$};
    \draw[dotted] (1,6) -- (2.6,6);
    \draw (1,6.1) -- (1,5.9) node[below] {$0$};
    \filldraw [black] 
     (2.6,6) circle (2pt) node[above] {$R_4$} node[below] {$\frac85$};
    \filldraw [black] 
     (3.33,6) circle (2pt) node[above] {$G_2$} node[below] {$\frac73$};
    \end{tikzpicture}
\end{center}
	\caption{Top: Illustration of the tradeoff between the competitive and discovery ratios with respect to $\Sigma$ in $G_\alpha$. Here, each point corresponds to a value of $\alpha>1$. The bold point corresponds to the strategy $G_2$ and all strategies $R_t$. Bottom: The discover ratios with respect to $\Sigma_9$.}
	\label{fig:Ga_cr_br}
\end{figure}

\bibliography{bijective_ls}

% \section{not unique}
% \begin{theorem}
% \label{th:br_opt_not_unique}
% ... 
% \end{theorem}
% Consider a strategy $X=(x_0,x_1, \ldots)$, with $x_i = \bar x_i$ for $i \in [0, 4]$, and $x_i = $

%%
%% Bibliography
%%

%% Please use bibtex, 

\end{document}